\pdfoutput=1

\documentclass[11pt]{article}
\usepackage{fontawesome}
\usepackage{thm-restate}

\usepackage[]{eacl2023}

\usepackage{times}
\usepackage{latexsym}
\usepackage{enumerate}
\usepackage{enumitem}

\DeclareRobustCommand
  \Compactcdots{\mathinner{\cdotp\mkern-2mu\cdotp\mkern-2mu\cdotp}}

\makeatletter
\@ifpackageloaded{stix}{%
}{%
  \DeclareFontEncoding{LS2}{}{\noaccents@}
  \DeclareFontSubstitution{LS2}{stix}{m}{n}
  \DeclareSymbolFont{stix@largesymbols}{LS2}{stixex}{m}{n}
  \SetSymbolFont{stix@largesymbols}{bold}{LS2}{stixex}{b}{n}
  \DeclareMathDelimiter{\lBrace}{\mathopen} {stix@largesymbols}{"E8}%
                                            {stix@largesymbols}{"0E}
  \DeclareMathDelimiter{\rBrace}{\mathclose}{stix@largesymbols}{"E9}%
                                            {stix@largesymbols}{"0F}
}
\makeatother
\newcommand{\multiset}[1]{\lBrace #1 \rBrace}

\usepackage[T1]{fontenc}
\usepackage[utf8]{inputenc}

\usepackage{microtype}
\usepackage{times}
\usepackage{latexsym}
\usepackage{aflt}

\usepackage{amsmath}

\makeatletter
\newcommand\footnoteref[1]{\protected@xdef\@thefnmark{\ref{#1}}\@footnotemark} %
\makeatother

\newcommand{\saveForCr}[1]{}
\newcommand{\strz}{{\color{black} \boldsymbol{z}}}
\title{On the Intersection of Context-Free and Regular Languages}

\usepackage{tipa}
\newcommand{\ucambridge}{\normalfont \text{\textipa{D}}}

\newcommand{\ethz}{\text{\normalfont \textipa{Q}}}

\newcommand{\usi}{\normalfont \text{\textipa{N}}}

\newcommand{\jhu}{\normalfont \text{\textipa{6}}}

\author{%
Clemente Pasti$^{\usi,\ethz}$%
~\;~\;~Andreas Opedal$^{\ethz}$%
~\;~\;~Tiago Pimentel$^{\ucambridge}$ \\
\textbf{Tim Vieira}$^{\jhu}$
~\;~\;~\textbf{Jason Eisner}$^{\jhu}$~\;~\;~\textbf{Ryan Cotterell}$^{\ethz}$\\
    $^{\usi}$Universit{\`a} della Svizzera Italiana
   ~\;~\;~\;~$^{\ethz}$ETH Z{\"u}rich \\
  $^{\ucambridge}$University of Cambridge%
  ~\;~\;~\;~$^{\jhu}$Johns Hopkins University
   \\
\texttt{\href{mailto:clemente.pasti@usi.ch}{clemente.pasti@usi.ch}}%
  ~\;~ \texttt{\href{mailto:andreas.opedal@inf.ethz.ch}{andreas.opedal@inf.ethz.ch}}%
  ~\;~ \texttt{\href{mailto:tp472@cam.ac.uk}{tp472@cam.ac.uk}} \\  \texttt{\href{mailto:tim.f.vieira@gmail.com}{tim.f.vieira@gmail.com}}%
  ~\;~ \texttt{\href{mailto:jason@cs.jhu.edu}{jason@cs.jhu.edu}}%
  ~\;~ \texttt{\href{mailto:ryan.cotterell@inf.ethz.ch}{ryan.cotterell@inf.ethz.ch}}
 }

\date{}

\begin{document}
\maketitle
\begin{abstract}
The Bar-Hillel construction is a classic result in formal language theory. 
It shows, by a simple construction, that the intersection of a context-free language and a regular language is itself context-free.  
In the construction, the regular language is specified by a finite-state automaton.
However, neither the original construction \citep{BarHillel61} nor its weighted extension \citep{nederhof-satta-2003-probabilistic} can handle finite-state automata with $\varepsilon$-arcs. 
While it is possible to remove $\varepsilon$-arcs from a finite-state automaton efficiently without modifying the language, such an operation modifies the automaton's set of paths.
We give a construction that generalizes the Bar-Hillel in the case where the desired automaton has $\varepsilon$-arcs, and further prove that our generalized construction leads to a grammar that encodes the structure of both the input automaton and grammar while retaining the asymptotic size of the original construction.
\newline
\newline
\vspace{1.5em} 
\hspace{.5em}\includegraphics[width=1.25em,height=1.25em]{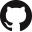}{\hspace{.75em}\parbox{\dimexpr\linewidth-2\fboxsep-2\fboxrule}{\url{https://github.com/rycolab/bar-hillel}}}
\end{abstract}

\vspace{-.75cm}
\section{Introduction}

\citeposs{BarHillel61} construction---together with its weighted generalization \citep{nederhof-satta-2003-probabilistic}---is a fundamental result in formal language theory.
Given a weighted context-free grammar (WCFG) $\grammar$
and a weighted finite-state automaton (WSFA) $\automaton$, the Bar-Hillel construction yields another WCFG $\grammarcap$ whose language $L(\grammarcap)$ is equal to the intersection of $\grammarlanguage$ with $\automatonlanguage$.
Importantly, the Bar-Hillel construction directly proves that weighted context-free languages are closed under intersection with weighted regular languages. The construction was later extended to other formalisms, e.g., tree automata \citep{maletti-satta-2009-parsing}, synchronous tree substitution grammars \citep{maletti-2010-synchronous} and linear context-free re-writing systems \citep{Kasami-Seki,nederhof-satta-2011-prefix-probabilities}.
Furthermore, the Bar-Hillel construction has seen applications in the computation of infix probabilities \citep{nederhof-satta-2011-computation} and human sentence comprehension \citep{levy-2008-noisy, levy-2011-integrating}.\looseness=-1

Unfortunately, \citeauthor{BarHillel61}'s construction, as well as its weighted generalization by \citeauthor{nederhof-satta-2003-probabilistic}, requires the input automaton to be $\varepsilon$-free.\footnote{But they do not require the input grammar to be $\varepsilon$-free.}
Although any WFSA can be converted to a weakly equivalent\footnote{Two WFSAs are said to be weakly equivalent if they represent the same weighted formal language. \label{fn:weak-equivalence}}
$\varepsilon$-free WFSA using well-known techniques \citep{Mohri2000GenericE, mohri_semiring, Hanneforth2010RemovalBL}, such an approach adds an additional step of computation, typically increases the size of the output grammar $\grammarcap$, and does not, in general, maintain a bijection between derivations in $\grammarcap$ and the Cartesian product of the derivations in $\grammar$ and paths in $\automaton$. 
In other words, $\grammarcap$ is \emph{not} strongly equivalent
to the 
product of $\grammar$ and $\automaton$.\footnote{Strong equivalence is formally defined in \cref{defn: weighted join} and \cref{thm: central theorem}.%
}\looseness=-1
 
In this note, we generalize the classical Bar-Hillel construction to the case where the automaton we seek to intersect with the grammar has $\varepsilon$-arcs.
Our new construction produces a WCFG $\grammarcap$
that is strongly equivalent to the product of $\grammar$ and $\automaton$.
We further generalize the Bar-Hillel construction to work with arbitrary commutative semirings.
Finally, we give an asymptotic bound on the size of the resulting grammar and a detailed proof of correctness in the appendix.
\looseness=-1

\section{Languages, Automata, and Grammars}
As background, we now give formal definitions of semirings, weighted formal languages, finite-state automata, and context-free grammars.

\begin{figure*}[t]
\centering
\tabskip=0pt
\valign{#\cr
  \noalign{\hfill}\vfill
  \hbox{%
    \begin{subfigure}{.5\textwidth}
    \centering
    \resizebox{.85\textwidth}{!}{%
\begin{tikzpicture}[node distance = 28mm]
    \node[state, initial] (q0) [] { $q_0$ };
    \node[state] (q1) [ right of=q0] { $q_1$ };
    \node[state] (q2) [ right of=q1] { $q_2$ };
    \node[state, accepting] (q3) [ right of=q2] { $q_3$ };
    \draw[-{Latex[length=3mm]}] 
    (q0) edge[ above] node{ $\textit{The}/ 2$ } (q1) 
    (q1) edge[ above ] node{ $\varepsilon/ 0.3$ } (q2)
    (q2) edge[ loop above  ] node{$\textit{many}/0.75$} (q2)
    (q2) edge[ above] node{ $\textit{cyclists}/ 1$ } (q3) 
    (q3) edge[ loop above] node{ $\varepsilon/ 0.6$ } (q3)
    ;
    
\end{tikzpicture}
}
    \caption{Weighted finite-state automaton}
    \end{subfigure}%
  }\vfill
  \hbox{%
    \begin{subfigure}{.5\textwidth}
    \centering
    \small
    \vspace{-3pt}
    \input{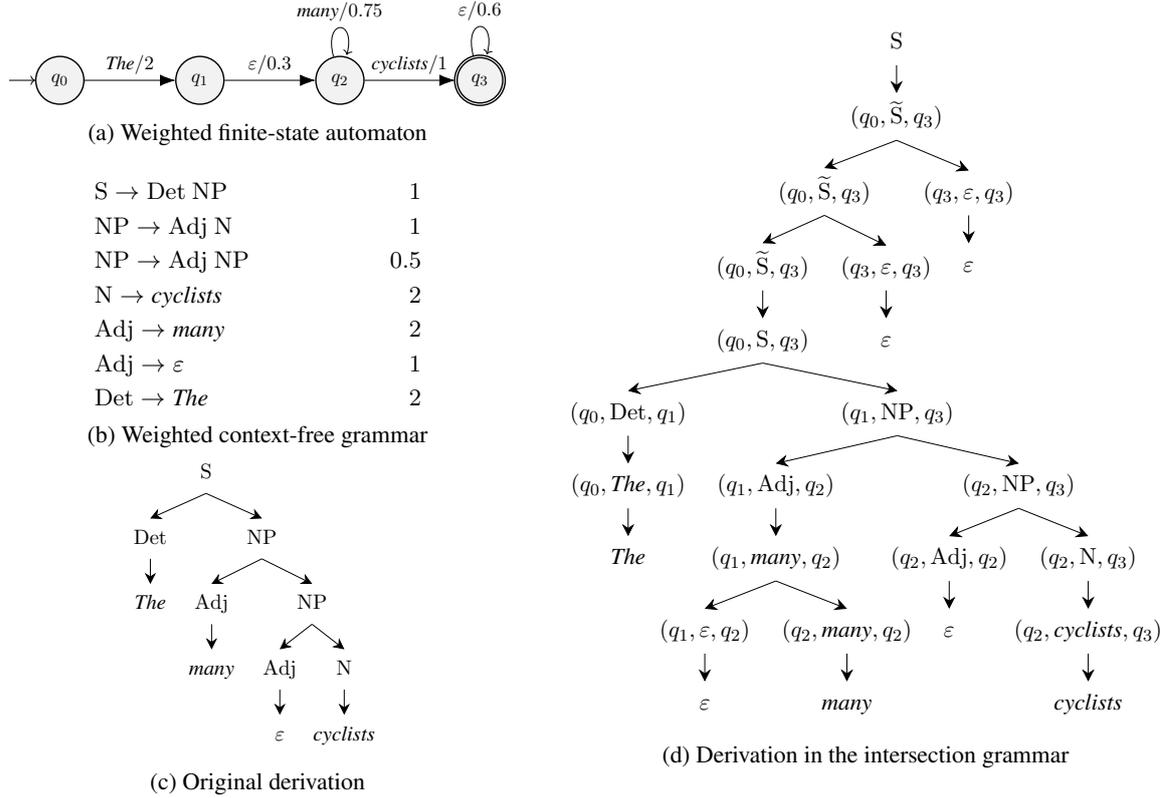}
    \vspace{-20pt}
    \caption{Weighted context-free grammar}
    \end{subfigure}%
  }\vfill
  \hbox{%
    \begin{subfigure}{.5\textwidth}
    \centering
    \resizebox{.42\textwidth}{!}{%
\begin{forest}
[,phantom,s sep=1.5cm,    
        [ $\NT{S}$
                [ $\NT{Det}$ [  $\textit{The}$ ] ]             
                [  $\NT{NP}$ 
                    [   $\NT{Adj}$    [ $\textit{many}$]  ]
                    [ $\NT{NP}$ 
                        [$\NT{Adj} $ [$\varepsilon$] ]
                        [$\NT{N}$ [$\textit{cyclists}$]  ]
                    ]
                ]
        ]
]
\end{forest}%
}
    \caption{Original derivation}
    \end{subfigure}%
  }\vfill\cr
  \noalign{\hfill}\vfill
  \hbox{%
    \begin{subfigure}[b]{.5\textwidth}
    \centering
    \resizebox{\textwidth}{!}{%
\begin{forest}
[,phantom,s sep=1.5cm,
    [$\NT{S}$ 
        [$\fsanonterm{\negterm{S}}{q_0}{q_3}$ 
            [ $\fsanonterm{\negterm{S}}{q_0}{q_3}$
                [ $\fsanonterm{\negterm{S}}{q_0}{q_3}$ 
                    [ $\fsanonterm{\NT{S}}{q_0}{q_3}$
                        [ $\fsanonterm{\NT{Det}}{q_0}{q_1}$ 
                            [ $\fsaterm{\textit{The}}{q_0}{q_1}$ [$\textit{The}$]] 
                        ]
                        [  $\fsanonterm{\NT{NP}}{q_1}{q_3}$ 
                            [   $\fsanonterm{\NT{Adj}}{q_1}{q_2}$
                                [$\fsaterm{\textit{many}}{q_1}{q_2}$ 
                                    [$\fsaterm{\varepsilon}{q_1}{q_2}$ 
                                        [$\varepsilon$]
                                    ] 
                                    [$\fsaterm{\textit{many}}{q_2}{q_2}$ 
                                        [$\textit{many}$]
                                    ]
                                ] 
                            ]
                            [ $\fsanonterm{\NT{NP}}{q_2}{q_3}$ 
                                [$\fsanonterm{\NT{Adj}}{q_2}{q_2}$ 
                                    [$\varepsilon$]
                                ] 
                                [$\fsanonterm{\NT{N}}{q_2}{q_3}$
                                    [$\fsaterm{\textit{cyclists}}{q_2}{q_3}$ 
                                        [$\textit{cyclists}$]
                                    ]
                                ]
                            ]
                        ]
                    ]
                ]
                [ $\fsanonterm{\varepsilon}{q_3}{q_3}$ [ $\varepsilon$ ] ]
            ]
            [ $\fsanonterm{\varepsilon}{q_3}{q_3}$ 
                [ $\varepsilon$ ] 
            ]
        ]
    ]
]
\end{forest}%
}
    \caption{\label{fig:example_intersection_derivation}Derivation in the intersection grammar}
    \end{subfigure}%
  }\vfill\cr
  \noalign{\hfill}
}
\caption{Example of a derivation in the grammar obtained as the intersection of the finite-state automaton (a) and the context-free grammar (b). 
The derivation tree (d) encodes the derivation tree (c) in the original grammar, and path $\edge{q_0}{\textit{The}}{2}{q_1}\edge{}{\varepsilon}{0.3}{q_2}\edge{}{\textit{many}}{0.75}{q_2} \edge{}{\textit{cyclists}}{1}{q_3} \edge{}{\varepsilon}{0.6}{q_3} \edge{}{\varepsilon}{0.6}{q_3}$.
We use rules from \cref{eqn:paired_2_BHI_automaton_epsilon} for $\epsilon$-arcs appearing before an input symbol, and rules from \cref{eqn:paired_2_BHI_epsilon_other} for $\epsilon$-arcs appearing at the end of the input.}
\label{fig:tree}
\vspace{-7pt}
\end{figure*}

\subsection{Semirings}
\label{sec:semirings}
Semirings
are useful algebraic structures for describing weighted languages \citep[Chapter 1]{AutomataHandbook}. In order to define semirings we must first give the definition of a monoid.
A \defn{monoid} is a 3-tuple $\mathcal{M}=(\semiringset, \bullet,\one)$, where $\semiringset$ is a set, $\bullet: \semiringset \times \semiringset \rightarrow \semiringset$ is an associative operator, and $\one \,\in\, \semiringset$ is a distinguished identity element such that $\one \bullet w = w\bullet \one = w$ for any $w \in \semiringset$. We say that a monoid is \emph{commutative} if $\bullet$ commutes, i.e., $w_1 \bullet w_2 = w_2 \bullet w_1$ for any $w_1,w_2 \in \semiringset$. 
We can now give the definition of a semiring.
\begin{defin}
A \defn{semiring} $\semiring=\semiringtuple$ is a 5-tuple where $(\semiringset,\oplus,\zero)$ is a commutative monoid,  $(\semiringset,\otimes,\one)$ is a monoid, $\otimes$ distributes over $\oplus$, and $\zero$ is an \emph{annihilator} for $\otimes$, meaning that $\zero \otimes w = w \otimes \zero =\zero$ for any $w \in \semiringset$.
\end{defin}
We say that $\mathcal{W}$ is commutative if $\otimes$ commutes. In this work, we assume commutative semirings.

\subsection{Weighted Formal Languages}

This paper concerns itself with transforms between devices that generate
weighted formal languages.
\begin{defin}\label{def:formal-language}
Let $\alphabet$ be an alphabet and $\semiring=\semiringtuple$ be a semiring. 
Then a \defn{weighted formal language} $L: \alphabet^* \rightarrow \semiringset$ is a  mapping from the Kleene closure of $\alphabet$ to the set of weights $\semiringset$.
Furthermore, the set $\support{L}= \big \{  \str \in \alphabet^* \mid L(\str) \neq \zero \big \}$ is called the language's \defn{support}.\looseness=-1
\end{defin}
Unweighted formal languages \citep[e.g.,][]{sipser2006,hopcroft_and_ullmann} are simply the special case of \cref{def:formal-language} where $\semiring$ is the boolean semiring. 
In this note, we discuss algorithms for computing the intersection of two weighted formal languages.\footnote{The intersection of two weighted languages is also called their Hadamard product \citep[Chapter 1]{AutomataHandbook}.}\looseness=-1
\begin{defin}\label{def:intersection}
Let $L_1$ and $L_2$ be two weighted formal languages over the same alphabet $\alphabet$ and the same semiring $\semiring$. 
The \defn{intersection} of $L_1$ with $L_2$ is defined as the weighted language
\begin{align}
    \left(L_1 \cap L_2\right) (\str)\defeq L_1(\str) \otimes L_2(\str),  \hspace{3mm}\forall \str \in \alphabet^*
\end{align}
\end{defin}
Specifically, this paper concerns
itself with the special case of \cref{def:intersection} when $L_1$ is a weighted context-free language (represented by a WCFG), and $L_2$ is a weighted regular language (represented by a WFSA); we define these two formalisms in the subsequent sections.

In the following, the symbol $\varepsilon$ always represents the empty string.

\subsection{Weighted Finite-State Automata}
\label{sec:wfsa}
We now review the basics of weighted finite-state automata (WFSA), which provide a formalism to represent weighted regular languages.
\begin{defin}
A \textbf{weighted finite-state automaton} $\automaton$ over a semiring $\semiring=\semiringtuple$ is a 6-tuple $\wfsatuple$.
In this tuple, $\alphabet$ is an alphabet, $\states$ is a finite set of states, and $\trans \subseteq \states \times \states \times (\alphabet\cup \{ \epsilon\})  \times \semiringset$ is a finite multi-set of weighted arcs.
Further, $\initweight : \states \rightarrow \semiringset$ and $\finalweight : \states \rightarrow \semiringset$ are the initial and final weight functions, respectively.
We also define the sets $\initstates = \{\state \mid \state \in \states,\, \initweight(\state) \neq \zero \}$ and $\finalstates = \{ \state \mid \state \in \states,\, \finalweight(\state) \neq \zero \}$ for convenience.
\end{defin}

We will represent an arc in $\trans$ with the notation $\edge{\state_0}{a}{\weight}{\state_1}$ where $a \in \alphabet\cup\{\varepsilon\}$ 
and $\weight \in \semiringset$.
A \textbf{path} $\apath$ (of length $N > 0$) is a sequence of arcs in $\trans^*$ where the states of adjacent arcs are matched, i.e.,%
\begin{align}
    q_0 \xrightarrow{a_1 / w_1} \Compactcdots q_{n-1}\xrightarrow{a_n / w_n} q_{n}\Compactcdots\xrightarrow{a_{N} / w_N} q_N
\end{align}
and where $q_0 \in \initstates$ and $q_N \in \finalstates$,
i.e., 
the path starts at an initial state and ends at a final state.
The path's \textbf{yield}, denoted $\yield{\apath}$, is the concatenation $a_1 a_2\Compactcdots a_N$ of all its arc labels (strings of length $\leq 1$).  The path's \textbf{weight}, denoted $\weight{\apath}$, is the product
\begin{equation}
\weight{\apath} = \initweight(q_0) \otimes \left( \bigotimes_{n=1}^N w_n \right) \otimes \finalweight(q_N)
\end{equation}
We denote the set of all paths in $\automaton$ as $\derivation{\automaton}$, and the set of all paths with yield $\str$ as $\derivation{\automaton}(\str)$.
Finally, we define the \defn{language of an automaton} as the mapping $L_{\automaton}: \alphabet^* \rightarrow \semiringset$ where we have\footnote{\label{fn:infsum}In the main paper we gloss over the question of how $\bigoplus$-summations over \emph{infinite} sets are to be defined (or left undefined), but we treat this issue in \cref{sec:proof-weights}.}
$ L_{\automaton}(\str) = \bigoplus_{\apath \in \derivation{\automaton}(\str)}  \weight{\apath}.$
The set of languages that can be encoded by a WFSA forms the class of \defn{weighted regular languages}.\looseness=-1

\subsection{Weighted Context-Free Grammars}\label{sec:wcfg}
We now go over the necessary background on weighted context-free grammars (WCFGs).\looseness=-1
\begin{defin}
A \textbf{weighted context-free grammar} is a tuple $\grammar = \wcfgtuple{}$, where $\nonterm$ is a non-empty set of nonterminal symbols, $\alphabet$ is an alphabet of terminal symbols, $\semiring=\semiringtuple$ is a semiring, $\start \in \nonterm$ is a distinguished start symbol, and $\rules$ is a set of production rules.
Each rule $\arule \in \rules$ is of the form $\wproduction{\nt{X}}{\valpha}{w}$, with $\nt{X} \in \nonterm$, $w \in \semiringset$, and $\valpha \in  (\alphabet \cup \nonterm)^*$.\looseness=-1
\end{defin}

Given two strings $\valpha, \vbeta \in (\alphabet \, \cup \, \nonterm)^*$,
we write $\valpha \overset{\arule}{\Rightarrow}_{L} \vbeta$ if and only if we can express $\valpha = \strz\,\NT{X}\,\vdelta$ and $\vbeta = \strz\,\vgamma\,\vdelta$ where $\strz \in \alphabet^*$ and  $\arule \in \rules$ is the rule $\wproduction{\NT{X}}{\vgamma}{w}$.
A \defn{derivation} $\tree$ (more precisely, a leftmost derivation) is a sequence $\valpha_0,\ldots,\valpha_{N}$ with $N > 0$, 
$\valpha_{0} = \start$, and $\valpha_N \in \alphabet^*$, such that for all $0 < n \leq N$, we have $\valpha_{n-1} \overset{\arule_n}{\Rightarrow}_{L}\valpha_{n}$ for some (necessarily unique) 
$\arule_n \in \rules$.
The derivation's \defn{yield}, $\yield{\tree}$, is $\valpha_{N}$, and its weight, $\weight(\tree)$, is  $\weight{\arule_1}\otimes\cdots\otimes\weight{\arule_N}$.
We denote the set of derivations under a grammar $\grammar$ as $\derivation_{\grammar}$ and the set of all derivations with yield $\str$ as $\derivation_{\grammar}(\str)$.
Finally, we define the \defn{language of a grammar} as $L_{\grammar}$ where\footnoteref{fn:infsum}\
$  L_{\grammar} (\str) \defeq \bigoplus_{\tree \in \derivation{\grammar}(\str)} 
    \weight{\tree}, \; \forall \str \in \alphabet^*$.
The languages that can be encoded by a WCFG are known as \defn{weighted context-free languages}.
\looseness=-1

\newcommand{\includevspacecompensation}{-30pt}

\section{Generalizing Bar-Hillel}\label{sec:generalizing}
Given any context-free grammar (CFG) $\grammar$ and finite-state automaton (FSA) $\automaton$, \citet{BarHillel61} showed how to construct a CFG $\grammarcap$ such that $L_{\grammarcap}\! =\! L_{\grammar} \cap L_{\automaton}$.  Later, \citet{nederhof-satta-2003-probabilistic} generalized Bar-Hillel's construction to work on a \emph{weighted} context-free grammar
and a \emph{weighted} finite-state automaton.  While they focused on the real semiring, their construction actually works for any commutative semiring. However, neither of these versions correctly computes the intersection when the WFSA (or FSA) contains $\varepsilon$-arcs.
Yet, in several applications---such as modeling noisy inputs for human sentence comprehension \citep{levy-2008-noisy, levy-2011-integrating}---we may be interested in using a WFSA $\automaton$ that contains $\varepsilon$-arcs. 
A na\"ive application of the construction would ignore paths in $\automaton$ that contain $\varepsilon$-arcs.
The problem may be sidestepped by transforming $\automaton$ into a weakly equivalent $\varepsilon$-free WFSA\footnote{See  \cref{fn:weak-equivalence} for the definition of weak equivalence.%
} before applying the construction;\footnote{\Citet{levy-2008-noisy, levy-2011-integrating} uses WFSAs to model the degree of uncertainty under which a human comprehends a particular sentence, in which $\varepsilon$-arcs are used to represent word deletion. He applies the Bar-Hillel construction to compute the intersection of the language represented by the WFSA and the language encoded by a WCFG that represents the comprehender's grammatical knowledge, in order to obtain a joint posterior distribution over parses and words. While he transforms $\automaton$ to eliminate $\varepsilon$-arcs prior to applying the Bar-Hillel construction (Levy, p.c.), the solution we propose here is an alternative.
}
this, however, might increase the size of the WFSA and of the intersection grammar, and it would not allow us to identify the paths in the input WFSA that yield a target string in the intersection grammar.\footnote{In contrast, this is easy under our construction. 
Each derivation of the target string under $\grammarcap$ uses a particular path in $\automaton$.
To reconstruct that path, $\varepsilon$-arcs and all, simply traverse from left to right the leaves of the derivation tree (e.g., \cref{fig:example_intersection_derivation}) and list the states on the triplets where rule \labelcref{eq:paired_2_fsa} is applied.}\looseness=-1
\newcommand{\darkredcolor}[1]{\textcolor{darkred}{#1}}
\newcommand{\lightredcolor}[1]{\textcolor{lightred}{#1}}
\newcommand{\darkgraycolor}[1]{\textcolor{darkgray}{#1}}
\newcommand{\newconstructioncolor}[1]{\textcolor{darkblue}{#1}}
\begin{figure*}[h]
    \centering
    \begin{subfigure}[t]{0.48\textwidth}
        \input{imgs/construction_1}
    \end{subfigure}%
    \hfill
    \begin{subfigure}[t]{0.48\textwidth}
        \input{imgs/construction_2}
    \end{subfigure}
    \caption{The original Bar-Hillel construction (left) and our generalized version (right) that covers $\varepsilon$-arcs. We highlight the differences from the original construction in \darkredcolor{red}. 
    Note that the weights of rules \labelcref{eq:paired_1_start,eq:paired_1_fsa} (respectively \labelcref{eq:paired_2_start,eq:paired_2_fsa}) encode the weights of the WFSA, while the weights of rules \labelcref{eq:paired_1_cfg,eq:paired_1_eps} (respectively \labelcref{eq:paired_2_cfg,eq:paired_2_eps}) encode weights of the WCFG. All other rules in the generalized construction (\labelcref{eqn:paired_2_BHI_automaton_epsilon,eqn:paired_2_BHI_exit_other,eqn:paired_2_BHI_epsilon_other}) are assigned weight $\one$, and, thus, they do not change the weight of a derivation.}
    \label{fig: figure new construction}
\end{figure*}

\subsection{The problem with $\varepsilon$-arcs}\label{sec:eps-arcs-problem}

Before proposing our solution, we explain how the original construction works, and how it fails in the case of $\varepsilon$-arcs. Given a  WFSA $\automaton=\wfsatuple$ and a WCFG $\grammar=\wcfgtuple{}$ over the same alphabet $\alphabet$ and commutative semiring $\semiring$, their intersection $\grammarcap$ is defined by the tuple $\grammarcaptuple$, where: 
\begin{itemize}[nosep]
\item The set of nonterminal symbols $\nontermcap=\{\start\}\cup \states\times( \nonterm \cup \alphabet )\times \states$ contains the triplets $\fsanonterm{\nt{X}}{q_i}{q_j}$ plus the start symbol $\start$.\footnote{\label{fn:useless}Many of the nonterminals will turn out to be \defn{useless} in that they do not participate in any derivation in $\derivation_{\grammarcap}$.  These can be pruned from the grammar along with all rules that mention them \cite{hopcroft_and_ullmann}.\saveForCr{We remark without details that as an optimization, it is possible to avoid constructing many of the useless nonterminals in the first place, by instantiating useful rules bottom-up and left-to-right---essentially running Earley's algorithm to parse the ``input'' $\automaton$.}}
\item The set of production rules $\rulescap$ is given by the equations in Construction 1 of \cref{fig: figure new construction}.\footnote{Note that this construction can handle multiple initial and final states, whereas \citeposs{nederhof-satta-2003-probabilistic} construction assumes a WFSA with a single initial and a single final state. A path's initial and final weights are taken into account by the weight of rules \labelcref{eq:paired_1_start} of Construction 1 in \cref{fig: figure new construction}.\looseness=-1}
\item $\alphabet$, $\semiring$, $\start$ are the same as in the input grammar.
\end{itemize}
The intuition behind this construction is that a derivation in the intersection grammar encodes both a path in the input WFSA and a derivation in the input WCFG with matching yield. Specifically, rules \labelcref{eq:paired_1_fsa} encode arcs in the WFSA and rules \labelcref{eq:paired_1_cfg} encode production rules in the WCFG. Rules \labelcref{eq:paired_1_eps} handle the special case of $\varepsilon$-productions in the input WCFG and rules \labelcref{eq:paired_1_start} are designed to take into account the initial and final weight of a path.
These rules may combine through matching nonterminals to permit derivations in the intersection grammar $\grammarcap$.  

Unfortunately, this mechanism breaks in the presence of $\varepsilon$-arcs.  Although the rules \labelcref{eq:paired_1_fsa} do construct nonterminals for $\varepsilon$-arcs (when $a=\varepsilon$), the  rules \labelcref{eq:paired_1_cfg} never generate those nonterminals (since the $\nt{X}_m$ on the right-hand side of a rule are never $\varepsilon$).  We show this with an example.
Consider the automaton and the grammar in \cref{fig:tree},  both of which assign non-zero weight to the string \textit{The many cyclists}. However, their intersection computed with the Bar-Hillel construction is empty. To see this, note that all the paths from $q_0$ to $q_3$ contain the arc $\edge{q_1}{\varepsilon}{0.3}{q_2}$. \cref{eq:paired_1_fsa} will create a rule $\wproduction{\fsanonterm{\varepsilon}{q_1}{q_2}}{\varepsilon}{0.3}$, but none of the rules produced by \cref{eq:paired_1_cfg,eq:paired_1_eps} has the triplet $\fsanonterm{\varepsilon}{q_1}{q_2}$ on the right hand side.
This misalignment results in an empty set of derivations in $\grammarcap$. In \cref{sec:appendix_failure_case} we describe more failure cases in a detailed manner.\looseness=-1
\subsection{Our generalized construction} 
We now 
describe an improved version of the Bar-Hillel construction that handles $\varepsilon$-arcs in the WFSA.
In comparison to the original construction, our version of $\grammarcap=\grammarcaptuple$ has
\begin{itemize}[nosep]
\item $\nontermcap= \{\start\}\, \cup \,\states \times (\nonterm \cup \{\negterm{\start} \}\cup \alphabet) \times \states $ as the set of nonterminals, where $\negterm{\start}$ is a new symbol;
\item $\rulescap$ as the augmented set of production rules given in Construction 2 of \cref{fig: figure new construction}.
\end{itemize}

\noindent
Our generalized construction adds additional production rules that traverse the $\varepsilon$-arcs. Rules \labelcref{eqn:paired_2_BHI_automaton_epsilon} can traverse a WFSA subpath labeled with $\varepsilon^* a$
to yield a terminal symbol $a \in \Sigma$.  At the end of the yielded string, rules \labelcref{eqn:paired_2_BHI_epsilon_other} can traverse a WFSA subpath 
labeled with $\varepsilon^*$ that ends at a final state $q_F$.
Our construction carefully avoids overcounting\footnote{\label{fn:nospurious}As \cref{fig:example_intersection_derivation} illustrates, we do this by introducing a single, right-branching subderivation for each $\automaton$-subpath $\varepsilon^* a$ that matches an input symbol $a$. A nonterminal of the form $(q_0,\varepsilon,q_1)$ is never used as a right child, nor does it ever combine with a nonterminal of the form $(q_0,\NT{X},q_1)$, except at the end of the input, which is specially handled by rules \labelcref{eqn:paired_2_BHI_epsilon_other}.  
Similarly, \citet{transducer_composition} avoid overcounting when intersecting or composing finite-state machines that have $\varepsilon$-arcs.\looseness=-1} 
by ensuring that each matching pair of an $\automaton$-path and a $\grammar$-derivation of its string corresponds to \emph{exactly one} $\grammarcap$-derivation of that string, as illustrated in \cref{fig:tree}.
Note that rules \labelcref{eq:paired_2_cfg,eq:paired_2_fsa,eq:paired_2_eps} are identical to their counterparts in the original construction.
Rules \labelcref{eq:paired_2_start} are a modified version of rules \labelcref{eq:paired_1_start} with the special start symbol $\negterm{\start}$; this allows our construction to handle $\varepsilon$-arcs immediately before the final state---by repeated applications of rule \labelcref{eqn:paired_2_BHI_epsilon_other}---before switching  $\negterm{\start}$ back to $\start$ with rule $\labelcref{eqn:paired_2_BHI_exit_other}$. 
In \cref{sec:appendix_failure_case} we illustrate the mechanism with examples.
\looseness=-1

We now state the theorem of correctness.\looseness=-1
\begin{defin}
\label{defn: weighted join}
Let $\alphabet$ be an alphabet and $\semiring$ be a commutative semiring.
Let $\grammar$ be a WCFG and $\automaton$ be a WFSA---both over $\alphabet$ and $\semiring$.
The \defn{weighted join}
of the derivations in $\derivation_{\grammar}$ with the paths in  $\derivation_{\automaton}$ is defined as:
\begin{align}
\left( \derivation_{\grammar} \bowtie \derivation_{\automaton}\right) \defeq \Big\{ & \langle \tree, \apath\rangle \mid \tree \in \derivation_{\grammar}, \apath \in \derivation_{\automaton} \\ 
&\text{ s.t. } \yield(\tree) = \yield(\apath)  \Big\} \nonumber
\end{align}
with $\weight\left(\langle \tree, \apath\rangle \right) = \weight(\tree) \otimes \weight(\apath)$.
\end{defin}
\begin{restatable*}{theorem}{maintheorem}
\label{thm: central theorem}
Let $\grammar$ be a WCFG and $\automaton$ a WFSA over the same alphabet $\alphabet$ and commutative semiring $\semiring$. 
Let $\grammarcap$ be the grammar obtained with our generalized construction. 
Then we have
strong equivalence between $\grammarcap$ and $\langle\grammar, \automaton\rangle$; meaning that there is a weight-preserving, yield-preserving bijection between $\derivation{\grammarcap}$ and  $\left( \derivation_{\grammar}\bowtie\derivation_{\automaton}\right)$.
\end{restatable*}
\begin{restatable*}{corollary}{maincorollary}\label{thm:weights}
$\grammarcap$ and $\langle\grammar, \automaton\rangle$ are weakly equivalent, meaning that $L_{\grammarcap}(\str) = L_{\grammar}(\str) \otimes L_{\automaton}(\str)$ whenever the values on the right-hand side are defined.
\end{restatable*}
See \cref{sec:big-proof} for proofs.
\cref{thm: central theorem} may be seen as a generalization of Theorem $8.1$ by \citet{BarHillel61} and Theorem $12$ by \citet{nederhof-satta-2003-probabilistic}. 
Indeed, the set of derivations produced by Construction 1 is equivalent to the set of derivations produced by Construction 2, modulo an unfold transform \citep{tamaki-sato-unfold-fold} to remove rules containing $\negterm{\start}$.
 Among the groups of rules listed in \cref{fig: figure new construction}, the set of rules with maximum cardinality
is the one defined by \cref{eq:paired_2_cfg}. This set has cardinality $\bigo{|\rules||\states|^{M_\star}}$, where $M_\star$ is $1$ plus the length of the longest right-hand side
among all the rules $\rules$. 
All other equations in this construction 
lead to smaller sets of added rules. Since \cref{eq:paired_2_cfg} is unchanged from \cref{eq:paired_1_cfg} in the original construction, the asymptotic bound on the number of rules in our output grammar remains unchanged.\looseness=-1
\section{Conclusion}
We generalized the weighted Bar-Hillel intersection construction so that the given WFSA may contain $\varepsilon$-arcs.
Our construction is strongly equivalent to the product of the original WCFG and WFSA, i.e., every derivation tree in the resulting grammar represents a pairing of a derivation tree in the input WCFG and a path in the WFSA with the same yield.  We gave a full proof of correctness for our construction.  By adding output strings to the WFSA arcs and having rule \labelcref{eq:paired_2_fsa} rewrite to the arc's output string, our method can also be used to compose a WCFG with a weighted finite-state \emph{transducer} (WFST) that could usefully model morphological post-processing or speaker errors.
\looseness=-1

\section{Acknowledgements}
The authors acknowledge Roger Levy for correspondence about \citet{levy-2008-noisy} and \citet{levy-2011-integrating}.
\section{Limitations}
In this note, we generalize a fundamental theoretical result in formal language theory, which has seen a variety of practical applications, including human sentence comprehension under uncertain input \citep{levy-2008-noisy,levy-2011-integrating} and infix probability computation \citep{nederhof-satta-2003-probabilistic}.
Although we motivate our paper by discussing the necessity of performing intersections on automata with $\varepsilon$-arcs, we do not explore any such practical applications.
Further, while we show that the asymptotic bound on the size of our intersection grammar matches the original Bar-Hillel construction's, we do not discuss multiplicative or added constants introduced in our grammar's size.

\section*{Ethical Statement}
We do not foresee any ethical issues with our work.
\bibliographystyle{acl_natbib}
\bibliography{custom}

\appendix
\onecolumn

\section{Failure Cases of Original Construction}
\label{sec:appendix_failure_case}
We distinguish two types of failure cases:
(i) $\support{\lang_{\grammarcap}}\neq \support{\lang_{\automaton}}\cap\support{\lang_{\grammar}}$ and (ii) $\lang_{\grammarcap} \neq \lang_{\automaton}\cap\lang_{\grammar}$, both of which we will exemplify now. 
Notably, the case (ii) follows from (i), but---to be comprehensible---we will nonetheless give an example where (ii) fails without (i).
For case (i), consider the following unweighted FSA:
\begin{figure}[H] 
    \centering
     \begin{tikzpicture}[node distance = 20mm]
     \node[state, initial] (q0) [] { $q_0$ }; 
     \node[state] (q1) [ right of=q0] { $q_1$ };
     \node[state] (q2) [ right of=q1] { $q_2$ };
     \node[state, accepting] (q3) [ right of=q2] { $q_3$ };
     \draw[-{Latex[length=2.5mm]}]
     (q0) edge[ above] node{ $\textit{a}$ } (q1)
     (q1) edge[ above] node{ $\varepsilon$ } (q2)
     (q2) edge[ above] node{$\textit{b}$} (q3);
     \end{tikzpicture}
 \label{fig:failure_fsa_1}
\end{figure}
\noindent and the following unweighted CFG:
\begin{align*}
    &\production{\NT{\start}}{\NT{A} \, \nt{B}}  \\
    &\production{\NT{A}}{\textit{a}} \\
    &\production{\NT{B}}{\textit{b}}
\end{align*}
It is easy to see that the intersection of the language accepted by the FSA and the language generated by the CFG is $\{\textit{ab}\}$. 
Construction 1, however, outputs an empty grammar (after pruning useless rules as in \cref{fn:useless}) and, hence, an empty language. 
To see this, consider \cref{eq:paired_1_cfg} and \cref{eq:paired_1_fsa}. 
First, \cref{eq:paired_1_fsa} will create a rule $\production{\fsanonterm{\varepsilon}{\state_1}{\state_2}}{\varepsilon}$, but $\fsanonterm{\varepsilon}{\state_1}{\state_2}$ will be useless because it cannot be reached from any of the rules produced by \cref{eq:paired_1_cfg}. 
Second, \cref{eq:paired_1_cfg} will produce reachable nonterminals
$\fsanonterm{\nt{A}}{\state_0}{\state_i}$ and $\fsanonterm{\nt{B}}{\state_i}{\state_3}$, with $i\in\{1,2\}$. The case of $i=1$ will reach $\textit{a}$ but not $\textit{b}$, and $i=2$ will reach $\textit{b}$ but not $\textit{a}$. Let us now show how our generalized construction fixes this failure case. \cref{eqn:paired_2_BHI_automaton_epsilon} generates the rule $\production{\fsaterm{\textit{b}}{q_1}{q_3}}{\fsaterm{\varepsilon}{q_1}{q_2}\fsaterm{\textit{b}}{q_2}{q_3}}$ which then combines with rule $\production{\fsanonterm{B}{q_1}{q_3}}{\fsaterm{\textit{b}}{q_1}{q_3}}$ to form a subderivation\footnote{In \cref{sec:big-proof} we give a formal definition of subderivation.} that covers the substring $\varepsilon \textit{b}$, as shown in the picture below.

\begin{figure}[H]
    \centering
 \begin{forest}
[,phantom,s sep=1.5cm,
 [ $\fsanonterm{B}{q_1}{q_3}$
 [ $\fsaterm{\textit{b}}{q_1}{q_3}$
    [$\fsaterm{\varepsilon}{q_1}{q_2}$ [$\varepsilon$]]
    [$\fsaterm{\textit{b}}{q_2}{q_3}$ [$\textit{b}$]]
]
]
]
\end{forest}
\end{figure}
\noindent Note that rules generated by \cref{eqn:paired_2_BHI_automaton_epsilon} can only mention symbol $\varepsilon$ in the left child, not in the right child, as discussed in \cref{fn:nospurious}.

As stated above, to be comprehensive, we also show a case where only case (ii) fails, without (i). Take the following WFSA over the Inside semiring \citep{huang-2008-advanced}:
\begin{figure}[H] 
    \centering
    \begin{tikzpicture}[node distance = 27mm]
    \node[state, initial] (q0) [] { $q_0/ 1$ }; 
    \node[state] (q1) [ right of=q0] { $q_1$ };
    \node[state, accepting] (q2) [ right of=q1] { $q_2/ 1$ };
    \draw[-{Latex[length=2.5mm]}] 
    (q0) edge[ above] node{ $\textit{a}/ 1$ } (q1) 
    (q1) edge[ loop above] node{ $\varepsilon/ \frac{1}{3}$ } (q1)
    (q1) edge[ above] node{ $\textit{b}/ 1$ } (q2) 
    ;
    \end{tikzpicture}
\label{fig:failure_fsa_2}
\end{figure}
\noindent and the same grammar as above with weight $1$ for all rules. 
It is easy to see that the language's weight for $\str=\textit{ab}$ in the WFSA is a geometric series $\lang_{\automaton} (\str)=\sum_{i=0}^{\infty}\left(\frac{1}{3}\right)^i=\frac{3}{2}$, while in the WCFG, $\lang_{\grammar}(\str)=1$. However, the  output grammar $\grammarcap$ of Construction 1 will contain one single derivation $\tree$:

\begin{figure}[H]
\centering
 \begin{forest}
[,phantom,s sep=1.5cm, 
 [$\start$ 
    [$\fsanonterm{\start}{\state_0}{\state_2}$ 
    [$\fsanonterm{\nt{A}}{\state_0}{\state_1}$
        [$\fsanonterm{a}{\state_0}{\state_1}$ [$\textit{a}$]]
        ]
    [$\fsanonterm{B}{\state_1}{\state_2}$ 
        [$\fsanonterm{\textit{b}}{\state_1}{\state_2}$
        [$\textit{b}$]]
        ]
    ]
]
]
 \end{forest}
\end{figure}
\noindent with $\weight{\tree}=1$, as all rules either stem from $\grammar$ or from the arcs $\edge{\state_0}{\textit{a}}{1}{\state_1}$ and $\edge{\state_1}{\textit{b}}{1}{\state_2}$.This will result in $L_{\grammarcap}=1$, but $L_{\automaton} \cap L_{\grammar}=\frac{3}{2}$. 
This is because there are no derivations rooted at $\start$ in $\grammarcap$ that  match with the $\varepsilon$-arcs in $\automaton$: Similarly to the example above, $\fsanonterm{\varepsilon}{\state_1}{\state_1}$ will not be reachable.
We will now briefly show how our construction fixes this failure case as well. Note that there are infinitely many paths in the WFSA with yield $\str=\textit{ab}$; but there is also only a single derivation in $\derivation_{\grammar}$ with this yield. Our construction thus ensures that there is exactly one derivation in $\derivation_{\grammarcap}$ for every \textit{ab} path in $\derivation_{\automaton}$.  As the $\varepsilon$-loop allows unboundedly long subpaths from $q_1$ to $q_2$ that are labeled with $\varepsilon^* b$, the rules generated by \cref{eqn:paired_2_BHI_automaton_epsilon} will build corresponding unboundedly deep subderivations of the following form:
\begin{figure}[H]
\centering
 \begin{forest}
[,phantom,s sep=1.5cm, 
 [$\fsaterm{\textit{b}}{q_1}{q_2}$ 
    [$\fsaterm{\varepsilon}{q_1}{q_1}$ [$\varepsilon$]]
    [$\fsaterm{\textit{b}}{q_{1}}{q_2}$
    ,edge={red!0},edge label={node[midway,black,sloped,font=\Large] {$\cdots$}}
    [$\fsaterm{\varepsilon}{q_{1}}{q_{1}}$ [$\varepsilon$]]
    [$\fsaterm{\textit{b}}{q_{1}}{q_{2}}$ [$\textit{b}$]]]
]
]
 \end{forest}
\end{figure}

\noindent Finally we observe that a similar argument holds for rules generated by \cref{eqn:paired_2_BHI_epsilon_other}, and $\varepsilon$-arcs that occur immediately before a final state.

\section{Proofs}
\label{sec:big-proof}

\subsection{Proof of \cref{thm: central theorem}}

\cref{thm: central theorem} gives a result for derivations (which are always rooted at $\start$) and paths (which always connect an initial state with a final state). However, in order to prove this theorem we must also consider subderivations and subpaths. We define subderivations as follows: a \defn{subderivation} $\subtree$ is a sequence $\valpha_0,\ldots,\valpha_{N}$ with $N \geq 0$, where (i) in the case of $N > 0$,
$\valpha_{0} = \nt{X}$, $\nt{X}\in\nonterm$, and $\valpha_N \in (\varepsilon \cup \alphabet^*)$, such that for all $0 < n \leq N$, we have $\valpha_{n-1} \overset{\arule_n}{\Rightarrow}_{L}\valpha_{n}$ for some $\arule_n \in \rules$, and (ii) in the case of $N=0$, $\valpha_0\in\alphabet \cup \{ \varepsilon\}$. 
The weight and yield of subderivations are defined analogously to that of derivations. In the extended case of $N=0$, the yield is equal to $\valpha_0$ and the weight is set to $\one$. We will say that a subderivation is \defn{rooted} at $\NT{X}$ if $\valpha_{0}=\NT{X}$.
We denote the set of subderivations rooted at $\NT{X}$ with $\derivation_{\grammar}(\NT{X})$. 
Moreover, a subpath is defined as follows: A \textbf{subpath} $\asubpath$ (of length $N\geq 0$), is (i) in the case of $N > 0$, a sequence of arcs in $\trans^*$ where the states of adjacent arcs are matched, and (ii) in the case of $N=0$ a single state $\state \in \states$.\footnote{We note the difference to paths defined in \cref{sec:wfsa}: a subpath does not need to start in an initial state and end in a final state.\looseness=-1} The subpath's weight, denoted $\tweight{\asubpath}$, is the product
$\tweight{\asubpath} =  \bigotimes_{n=1}^N \weight_n$ of the weights of the arcs along the subpath. In the extended case $N=0$ we set the weight to $\one$ and the yield to $\varepsilon$. Note that, in contrast to the weight of a path, the weight of a subpath does not account for initial and final weights. The yield is defined identically to that of paths. We denote the set of all paths starting at $\state_i$ and ending at $\state_j$ with $\derivation{\automaton}(\{\state_i,\state_j\})$. Note that the definitions of subderivation and subpath encapsulate the definitions of derivation and path respectively.
Furthermore, we will denote with $\prevq{\apath}$ and $\nextq{\apath}$, respectively, the first and the last state encountered along a path.

We will now prove two lemmas that will be necessary for the proof of \cref{thm: central theorem}.

\definecolor{mygray}{RGB}{100,100,100}

\newcommand{\proofparagraph}[1]{\vspace{5pt} \noindent\textcolor{black}{\textit{#1}}}

\begin{lemma}
\label{lemma: main lemma}
For any triplet
$\fsanonterm{X}{\state_0}{\state_m} \in \nontermcap$, with $\NT{X}\neq \negterm{\start}$
and $q_0, q_m \in \states$, there is a bijection $\psi(\subtreecap)= \langle \subtree, \asubpath \rangle $ from $\derivation_{\grammarcap}\big(\fsanonterm{X}{\state_0}{\state_m}\big)$ to the weighted  join 
$(\derivation_{\grammar}(\NT{X}) \bowtie \derivation_{\automaton} (\{ q_0 , q_m \}) )$, restricted to tuples in which the path does not have an $\varepsilon$-arc immediately before a final state. 
Moreover, it holds that:
\begin{align}        &\weight{\subtreecap} 
        =\weight{\subtree}\otimes   \tweight{\asubpath} \label{eqn: lemma 1: weight} \\        &\yield{\subtreecap}=\yield{\subtree}=\yield{\asubpath} \label{eqn: lemma 1: yield}
        \end{align}
\end{lemma}
\begin{proof}
\noindent We begin by showing that $\psi$ is well defined, that it is injective and that it satisfies the properties in  \cref{eqn: lemma 1: yield,eqn: lemma 1: weight}. We prove this by induction on subderivations.

\proofparagraph{\cref{lemma: main lemma}'s Base Case.}
We begin by observing that the only terminal rules from $\rulescap$ are defined by \cref{eq:paired_2_fsa} and \cref{eq:paired_2_eps}. 

\proofparagraph{\cref{lemma: main lemma}'s Base Case, Part \#1.}
$\subtreecap$ is obtained by the application of a single production rule  $\wproduction{\fsaterm{a}{\state_0}{\state_1}}{a}{\weight}$ from \cref{eq:paired_2_fsa}. We define $\psi(\subtreecap)= \langle \subtree, \asubpath \rangle $, where $\asubpath= \edge{\state_0}{a}{w}{\state_1}$ and $\subtree=a$ is the subderivation that contains just the string $a$ with weight $\one$. It is easy to see that the yield is preserved. Moreover:\looseness=-1
\begin{subequations}
\begin{align}
\weight{\subtreecap}&=w &\mathcomment{by \cref{eq:paired_2_fsa}} \\
&= w \otimes \one \\
&= \tweight{\asubpath} \otimes \weight{\subtree}
\end{align}
\end{subequations}

\proofparagraph{\cref{lemma: main lemma}'s Base Case, Part \#2.}
$\subtreecap$ is obtained by the application of a single production rule $\wproduction{\fsanonterm{X}{\state_0}{\state_0}}{\varepsilon}{\weight}$ from 
\cref{eq:paired_2_eps}. We construct $\psi$ as follows: $\psi(\subtreecap)= \langle \subtree, \asubpath \rangle $, where $\subtree= \NT{X} \overset{\arule}{ \Rightarrow_{L}}  \varepsilon$ with  $\arule=\wproduction{\NT{X}}{\varepsilon}{w}$, and $\asubpath$ is the subpath $\state_0$ with weight $\one$. Clearly the yield is preserved and:
\begin{subequations}
\begin{align}
    \weight{\subtreecap}&=w  &\mathcomment{by \cref{eq:paired_2_eps}}\\
    &=w \otimes \one \\
    &=\weight{\subtree}\otimes \tweight{\asubpath}
\end{align}
\end{subequations}

\proofparagraph{\cref{lemma: main lemma}'s Induction Step.}
In the induction step, we show that the properties that we have shown for the base case propagate upwards along the derivation. 
In general, we will show that for any $\subtreecap= \fsanonterm{X}{\state_0}{\state_M} \overset{\arule}{\Rightarrow_{L}} \fsanonterm{X_1}{q_0}{q_1},\ldots,\fsanonterm{X_M}{q_{M-1}}{q_M} \Rightarrow_{L} \ldots$, we can construct $\psi(\subtreecap)= \langle \subtree, \asubpath \rangle $ such that the mapping is injective and that the properties in \cref{eqn: lemma 1: weight,eqn: lemma 1: yield} hold.
Additionally, as for the base case, we will show that $\asubpath$ connects $\state_0$ with $\state_M$ and that $\subtree$ is rooted at $\NT{X}$. 
As our inductive hypothesis, we will assume that each of these hypotheses hold for the subderivations rooted at each of the child nonterminals $\fsanonterm{X_1}{q_0}{q_1},\ldots,\fsanonterm{X_M}{q_{M-1}}{q_M}$.
We note that the  rules from $\rulescap$ which apply to a nonterminal of form $\fsanonterm{X}{\state_0}{\state_M}$  with $\nt{X}\in \alphabet$ are discussed in base case \#1, if instead  $\nt{X}\in \nonterm$, we either have base case \#2 or one of the rules defined by \cref{eq:paired_2_cfg} and \cref{eqn:paired_2_BHI_automaton_epsilon}; we discuss each now.

\proofparagraph{\cref{lemma: main lemma}'s Induction Step, Part \#1.} The topmost rule applied in $\subtreecap$ is $\arule=\wproduction{\fsanonterm{a}{q_0}{q_{2}}}{\fsanonterm{\varepsilon}{q_0}{q_{1}}\fsanonterm{a}{q_{1}}{q_{2}}}{\one}$ defined by \cref{eqn:paired_2_BHI_automaton_epsilon}.
We denote with $\subtree_{\cap,1}$ the subderivation rooted at $\fsanonterm{\varepsilon}{q_0}{q_{1}}$, and we observe that the only possible form for this derivation is
$\fsanonterm{\varepsilon}{q_0}{q_{1}} \overset{\arule}{\Rightarrow_{L}} \varepsilon$ for some $\arule=\wproduction{\fsanonterm{\varepsilon}{q_0}{q_{1}}}{\varepsilon}{w}$.
We denote with $\subtree_{\cap,2}$ the subderivation rooted at $\fsanonterm{a}{q_{1}}{q_{2}}$, then by inductive hypothesis, we know that there is a mapping $\psi(\subtree_{\cap,2})=\langle \subtree_{2},\asubpath_{2} \rangle $ such that \cref{eqn: lemma 1: weight,eqn: lemma 1: yield} are satisfied.

Then we construct $\psi(\subtreecap)= \langle \subtree, \asubpath \rangle $, so that $\subtree=\subtree_{2}$ and $\asubpath= \edge{q_0}{\varepsilon}{w}{q_1} \circ \asubpath_{2}$. As the yield of the subderivation rooted at $\fsanonterm{\varepsilon}{q_0}{q_1}$ is $\varepsilon$, the yield of $\subtreecap$ is the same as that of $\subtree_{\cap, 2}$. Further, the yield of $\asubpath$ is the same as $\asubpath_2$. We thus have that:
\begin{align}
    \yield{\subtreecap} = \yield{\subtree_{\cap,2}},\qquad 
    \yield{\subtree} = \yield{\subtree_{2}},\qquad 
    \yield{\asubpath} = \yield{\asubpath_2}
\end{align}
By induction, we have that the yield is preserved.
Similarly, we have that the weight is preserved:
\begin{subequations}
 \begin{align}
 \weight{\subtreecap}&= \one \otimes \weight{\subtree_{\cap,1}} \otimes \weight{\subtree_{\cap,2}} \\
 &= \one \otimes w \otimes \weight{\subtree_{2}} \otimes \tweight{\asubpath_{2}} & \mathcomment{inductive hypothesis} \\
 &= \weight{\subtree_{2}} \otimes \Big( w \otimes \tweight{\asubpath_{2}} \Big) &\mathcomment{commutativity} \\
 &= \weight{\subtree} \otimes \tweight{\asubpath}
 \end{align}
\end{subequations}
Finally, by induction we assume that $\asubpath_{2}$ connects state $q_{1}$ with state $q_{2}$, which implies that $\asubpath$ connects state $q_0$ with state $q_{2}$.

\proofparagraph{\cref{lemma: main lemma}'s Induction Step, Part \# 2.}
The topmost rule applied in $\subtreecap$ is $\arule= \wproduction{\fsanonterm{X}{q_0}{q_M}}{\fsanonterm{X_1}{q_0}{q_{1}}, \ldots, \fsanonterm{X_M}{q_{M-1}}{\state_M} }{w}$ defined by  \cref{eq:paired_2_cfg}. By induction we assume that the subderivation $\subtree_{\cap,m}$ rooted at $\fsanonterm{X_m}{q_{m-1}}{q_m}$ is mapped by $\psi$ into a subderivation $\subtree_m$ rooted at $\NT{X_m}$ and a path $\asubpath_m$, so that $\yield{\subtree_{\cap,m}}=\yield{\subtree_m}=\yield{\asubpath_m}$ and that $\weight{\subtree_{\cap,m}}=\weight{\subtree_m}\otimes\tweight{\asubpath_m}$. We then define $\psi ( \subtreecap ) = \langle \subtree, \asubpath \rangle $ where $\subtree= \NT{X} \overset{\arule}{\Rightarrow_{L}} \NT{X}_{1}, \ldots, \NT{X_M} \Rightarrow_{L} \ldots $ with $\arule= \wproduction{\NT{X}}{\NT{X_{1}}, \ldots \NT{X_M}}{w}$ and $\asubpath= \asubpath_{1} \circ \ldots \circ \asubpath_M$. As the states of neighboring triplets are matched, and by induction we assume  that $\asubpath_m$ connects  states $q_{m-1}$ with state $q_m$, we have that $\asubpath$ is a path from $q_0$ to $q_M$. We note that the yield of $\subtree$ is obtained by concatenation of  $\yield{\subtree_m}$ from left to right, and that similarly the yield of  $\asubpath$ is obtained by concatenation of $\yield{\asubpath_m}$ from left to right. 
This, together with the inductive hypothesis proves \cref{eqn: lemma 1: yield} of the lemma---as the yield of $\subtreecap$ will also be given by the concatenation of $\yield{\subtree_{\cap,m}}$ from left to right.
We now show that \cref{eqn: lemma 1: weight} on weights holds:
\begin{subequations}
\begin{align}
    \weight{\subtreecap}&= \weight \otimes \bigotimes_{m=1}^M \weight{\subtree_{\cap, m}} \\
    &= \weight \otimes \bigotimes_{m=1}^M \weight{\subtree_m} \otimes \tweight{\asubpath_m}
    &\mathcomment{inductive hypothesis}\\
    &= \Bigg( \weight \otimes \bigotimes_{m=1}^M \weight{\subtree_m} \Bigg) \otimes \bigotimes_{m=1}^M \tweight{\asubpath_m}
    &\mathcomment{commutativity} \\
    &= \weight{\subtree} \otimes \tweight{\asubpath}
\end{align}
\end{subequations}

We have defined  $\psi$ in a bottom-up fashion. At each step changing the topmost rule would result either in a different tree $\subtree$ or in a different path $\asubpath$, which proves injectivity. The proof that $\psi$ is surjective is very similar, and consists in showing by induction, that for any $\subtree \in \derivation_{\grammar}(\NT{X})$, and for any path $\asubpath$ that does not have a sequence of $\varepsilon$-arc before a final state, it is always possible to build a derivation in $\derivation_{\grammarcap}(\fsanonterm{\NT{X}}{\prevq{\apath}}{\nextq{\apath}})$. 
We limit ourselves to noting that it is always possible to do so by using rules from \cref{eq:paired_2_cfg,eq:paired_2_fsa,eq:paired_2_eps}, as in the original Bar-Hillel construction, and by using rules defined by \cref{eqn:paired_2_BHI_automaton_epsilon} to cover $\varepsilon$-arcs in the WFSA.

\end{proof}

\begin{lemma} \label{lemma: second lema}

For any triplet $\fsanonterm{\negterm{\start}}{\qinit}{q} \in \nontermcap$, with $\qinit \in \initstates, \state \in \states $, there is a bijection $\xi(\subtreecap)= \langle \subtree, \asubpath \rangle $ 
from $\derivation_{\grammarcap}\big(\fsanonterm{\negterm{\start}}{\qinit}{q}\big)$ to the join $( \derivation_{\grammar}(\start) \bowtie \derivation_{\automaton}(\{ \qinit, q\}))$, and we have that:
\begin{align}        \weight{\subtree_{\cap}} &=
        \weight{\subtree}\otimes   \tweight{\asubpath} \label{eqn: lemma 2 weight} \\        \yield{\subtree_{\cap}}&=\yield{\subtree}=\yield{\asubpath} \label{eqn: lemma 2 yield}
\end{align}
\end{lemma}

\begin{proof} We now present an inductive proof (similar to the above) for this lemma.

\proofparagraph{\cref{lemma: second lema}'s Base Case.}
The topmost rule applied in $\subtreecap$ is $\wproduction{\fsanonterm{\negterm{\start}}{\qinit}{q}}{\fsanonterm{\start}{\qinit}{q}}{\one}$ from rules defined by \cref{eqn:paired_2_BHI_exit_other}.
We denote with $\subtree_{\cap,1}$ the subderivation rooted at $\fsanonterm{\start}{\qinit}{q}$. Then by \cref{lemma: main lemma}, we know that there is a mapping $\psi(\subtree_{\cap,1})=\langle \subtree_1,\asubpath_1 \rangle $ such that \cref{eqn: lemma 2 weight,eqn: lemma 2 yield} are satisfied.
We then define $\xi(\subtree_{\cap})= \langle \subtree_1, \asubpath_1 \rangle $, and one can easily see that the properties in \cref{eqn: lemma 2 weight,eqn: lemma 2 yield} are satisfied.

\proofparagraph{\cref{lemma: second lema}'s Induction Step.} 
The topmost rule applied in $\subtreecap$ is $\wproduction{\fsanonterm{\negterm{\start}}{\qinit}{q_1} }{\fsanonterm{\negterm{\start}}{\qinit}{q_0} \fsanonterm{\varepsilon}{q_0}{q_1}}{\one}$ from rules defined by \cref{eqn:paired_2_BHI_epsilon_other}.
We denote with $\subtree_{\cap,1}$ the subderivation rooted at $\fsanonterm{\negterm{\start}}{\qinit}{q_0}$, and we assume by induction that $\xi(\subtree_{\cap,1})=\langle \subtree_1,\asubpath_1 \rangle $  and that properties in \cref{eqn: lemma 2 weight,eqn: lemma 2 yield} hold.
We denote with $\subtree_{\cap,2}$ the subderivation rooted at $\fsanonterm{\varepsilon}{q_0}{q_1}$, and we observe that the only possible form for this derivation is $\fsanonterm{\varepsilon}{q_0}{q_1} \overset{\arule}{\Rightarrow_{L}} \varepsilon$ for some $\arule=\wproduction{\fsanonterm{\varepsilon}{q_0}{q_1}}{\varepsilon}{w}$.
Then we can construct $\xi(\subtree_{\cap})= \langle \subtree, \asubpath \rangle $, where $\subtree=\subtree_1$ and $\asubpath=\asubpath_1 \circ \edge{q_0}{\varepsilon}{w}{q_1}$. 
The property in \cref{eqn: lemma 2 yield} is clearly satisfied, for property \cref{eqn: lemma 2 weight}, we have:
\begin{subequations}
\begin{align}
\weight{\subtreecap} &= \one \otimes \weight{\subtree_{\cap,1}} \otimes \weight{\subtree_{\cap,2}} \\
&= \weight{\subtree_{\cap,1 }} \otimes w &\mathcomment{weight of $\subtree_{\cap,2}$}\\
&=  \weight{\subtree_{1}} \otimes \tweight{\asubpath_{1}} \otimes  w &\mathcomment{inductive hypothesis} \\
&= \weight{\subtree} \otimes \tweight{\asubpath} &\mathcomment{weight of $\asubpath$}
\end{align}
\end{subequations}
As for \cref{lemma: main lemma} we note that modifying the topmost rule in $\subtreecap$, would always result either in a different derivation $\subtree$ or in a different path $\asubpath$, which proves injectivity. Surjectivity can be shown by induction, similarly to how we did for injectivity. We will simply note that given any derivation $\subtree$ rooted at $\start$, and given any path $\asubpath$ starting from an initial state, it is always possible to build a matching derivation $\subtreecap$in $\derivation_{\grammarcap}(\fsanonterm{\negterm{\start}}{\prevq{\apath}}{\nextq{\apath}})$, by using the result from \cref{lemma: main lemma}, and applying rules defined by \cref{eqn:paired_2_BHI_exit_other,eqn:paired_2_BHI_epsilon_other}.\looseness=-1    
\end{proof}

We can finally prove \cref{thm: central theorem}, which we restate here for convenience.

\maintheorem

\begin{proof}
Any derivation $\treecap$ in $\derivation_{\grammarcap}(\start)$ takes the form 
$\start \overset{\arule}{\Rightarrow_{L}} \fsanonterm{\negterm{\start}}{\qinit}{\qfinal} \Rightarrow_{L} \ldots $ with $\arule= \wproduction{\start}{\fsanonterm{\negterm{\start}}{\qinit}{\qfinal}}{\initweight(\qinit) \otimes \finalweight(\qfinal)}$, for $\qinit \in \initstates$ and $\qfinal \in \finalstates$.
 We denote with $\subtreecap$ the subderivation rooted at $\fsanonterm{\negterm{\start}}{\qinit}{\qfinal}$.
We can thus define $\phi(\treecap) = \langle \tree, \apath \rangle = \langle \subtree, \asubpath \rangle $, where $\xi(\subtreecap)= \langle \subtree, \asubpath \rangle $, and $\xi$ is the bijection defined in \cref{lemma: second lema}. 
By \cref{lemma: second lema}  we have that $\subtree = \tree $ is rooted at $\start$, and that $\asubpath=\apath$ has initial and final states:
$\prevq{\asubpath}=\qinit$ and $\nextq{\asubpath}=\qfinal$.
Clearly, $\yield{\tree_{\cap}}=\yield{\subtree_{\cap}}$ and, by \cref{lemma: second lema}, $\yield{\subtree_{\cap}}=\yield{\subtree}=\yield{\asubpath}$.
Further, by definition $\yield{\tree}=\yield{\subtree}$ and $\yield{\asubpath}=\yield{\apath}$.
Moreover, we have that:
\begin{subequations}
\begin{align}
    \weight{\treecap}&= \weight{\arule} \otimes \weight{\subtreecap} &\mathcomment{weight of a derivation} \\
    &=\weight{\arule} \otimes \weight{\subtree}  \otimes \tweight{\asubpath} &\mathcomment{\cref{lemma: second lema}}\\
    &= \initweight(\qinit) \otimes \finalweight(\qfinal) \otimes \weight{\subtree}  \otimes \tweight{\asubpath}  
    &\mathcomment{weight of $\arule$} \\
    &= \weight{\subtree}  \otimes \initweight(\qinit) \otimes \tweight{\asubpath} \otimes \finalweight(\qfinal) 
    &\mathcomment{commutativity} \\
    &= \weight{\tree} \otimes \weight{\apath} &\mathcomment{definition of weight of a path}
\end{align}
\end{subequations}
which proves that $\phi$ is weight and yield preserving. By \cref{lemma: second lema} we know that $\xi$ is a bijection, which implies that modifying the topmost rule $\arule$ would result in a different tuple $\langle \tree,\apath \rangle$. This proves the injectivity of $\phi$. Conversely, consider any path $\apath$ connecting an initial state with a final one and  any derivation $\tree$ rooted at $\start$, such that $\yield{\tree}=\yield{\apath}$. By \cref{lemma: second lema} we know that it is always possible to construct a subderivation $\subtreecap$, rooted at $\fsanonterm{\negterm{\start}}{\qinit}{\qfinal}$, that satisfies \cref{eqn: lemma 2 yield,eqn: lemma 2 weight}. Thus we can construct $\treecap= \start \overset{\arule}{\Rightarrow_{L}} \fsanonterm{\negterm{\start}}{\qinit}{\qfinal} \Rightarrow_{L} \cdots$ with $\arule=\wproduction{\start}{\fsanonterm{\negterm{\start}}{\qinit}{\qfinal}}{\initweight(\qinit) \otimes \finalweight (\qfinal)}$ a rule from \cref{eq:paired_2_start}. This shows the surjectivity of $\phi$. 

\end{proof}

\subsection{Proof of \cref{thm:weights}}\label{sec:proof-weights}

\maincorollary
\begin{proof}
\Cref{sec:semirings} defined both $L_{\automaton}(\str)$ and $L_{\grammar}(\str)$ as sums over derivations that yield $\str$.  If there are only finitely many such derivations, then the sum is well-defined by applying the associative--commutative operator $\oplus$ finitely many times.  However, \cref{fn:infsum} noted that countably infinite sums can arise.  We treat this issue by augmenting the semiring with an operator $\bigoplus$ that is applied to a countable (possibly infinite) multiset of weights and returns a value that is interpreted as the sum of those weights, or else returns a special ``undefined'' value $\bot \notin \semiringset$ to indicate that the sum diverges.

We require $\bigoplus$ to satisfy the following axioms for any two countable multisets $I, J \subseteq \semiringset$ such that
\begin{align}
\bigoplus I = W \in \semiringset \quad \quad \bigoplus J = V \in \semiringset
\end{align}
\begin{itemize}
\item \emph{Infinite distributivity}: Let $I \bigotimes J$ denote the multiset $\multiset{i \otimes j: i \in I, j \in J}$.  Then $\bigoplus (I \bigotimes J) = W \otimes V \in \semiringset$.  
\item \emph{Infinite associativity}: for any partition\footnote{Recall that partitions are definitionally disjoint.} $I = \bigcup_{k \in  K} I_k$, we have $\bigoplus I_k \in \semiringset$ for each $k \in K$ and furthermore $\bigoplus_{k \in K} \left(\bigoplus I_k\right) = W$.
\item \emph{Base cases}: For any $w,w' \in \semiringset$, $\bigoplus \multiset{w,w'} = w \oplus w'$, $\bigoplus \multiset{w} = w$, and $\bigoplus \multiset{} = \zero$.  Together with the previous property, this ensures that $\bigoplus$ agrees with the $\oplus$-based definition on finite multisets.
\end{itemize}
The first two axioms are adapted from part of \citet{mohri_semiring}'s definition of closed semirings.  The proof of \cref{thm:weights} uses only the first axiom, as follows.
Given a string $\str$ such that $L_{\automaton}(\str), L_{\grammar}(\str) \in \semiringset$.  
By definition (\crefrange{sec:wfsa}{sec:wcfg}), $L_{\automaton}(\str) = \bigoplus I$ and $L_{\grammar}(\str) = \bigoplus J$ if we define $I = \multiset{\weight{\apath}: \apath \in \derivation{\automaton}(\str)}$ and $J = \multiset{\weight{\tree}: \tree \in \derivation{\grammar}(\str)}$.  
Then also $L_{\grammarcap}(\str) = \bigoplus (I \bigotimes J)$ since $I \bigotimes J = \multiset{\weight{\tree}: \tree \in \derivation{\grammarcap}(\str)}$ according to \cref{thm: central theorem}.  By infinite distributivity, then, $L_{\grammarcap}(\str) = (\bigoplus I) \otimes (\bigoplus J) = L_{\automaton}(\str) \otimes L_{\grammar}(\str) \in \semiringset$ as claimed.
\end{proof}

\end{document}